\documentclass{ws-ijmpc}
\usepackage[T1]{fontenc}
\usepackage[latin9]{inputenc}
\usepackage{textcomp}
\usepackage{amsmath}
\usepackage{amssymb}
\usepackage{graphicx}
\usepackage[numbers]{natbib}

\makeatletter

\newcommand{\noun}[1]{\textsc{#1}}
\providecommand{\tabularnewline}{\\}

\usepackage{amsmath}

\DeclareMathOperator{\D}{D}

\DeclareMathOperator{\Df}{D\!}

\DeclareMathOperator{\car}{car}

\DeclareMathOperator{\sub}{sub}

\DeclareMathOperator{\csub}{csub}
\DeclareMathOperator{\bin}{bin\!}

\DeclareMathOperator{\und}{\mbox{undefined}}

\makeatother

\begin{document}
\markboth{Vojt\v{e}ch Vorel}{Subset Synchronization and Careful Synchronization of Binary Finite Automata}
\title{\uppercase{Subset Synchronization and Careful Synchronization of Binary Finite Automata}\footnote{Research supported by the Czech Science Foundation grant GA14-10799S and the GAUK grant No. 52215.}}

\author{\uppercase{Vojt\v{e}ch Vorel}}
\address{Department of Theoretical Computer Science and Mathematical Logic,\\
Charles University,\\ Malostransk\'{e} n\'{a}m. 25, Prague, Czech Republic\\ vorel@ktiml.mff.cuni.cz}

\maketitle

\begin{abstract} We present a strongly exponential lower bound that applies both to the subset synchronization threshold for binary deterministic automata and to the careful synchronization threshold for binary partial automata. In the later form, the result finishes the research initiated by Martyugin (2013). Moreover, we show that both the thresholds remain strongly exponential even if restricted to strongly connected binary automata. In addition, we apply our methods to computational complexity. Existence of a subset reset word is known to be  PSPACE-complete; we show that this holds even under the restriction to strongly connected binary automata. The results apply also to the corresponding thresholds in two more general settings: D1- and D3-directable nondeterministic automata and composition sequences over finite domains.
\keywords{Reset Word, Directing Word, Synchronizing Word, Composition Sequence, \v{C}ern\'{y} Conjecture} \end{abstract}

\section{Introduction}

Questions about synchronization of finite automata have been studied
since the early times of automata theory. The basic concept is very
natural: For a given machine, we want to find an input sequence that
would get the machine to a unique state, no matter in which state
the machine was before. Such sequence is called a \emph{reset word}%
\footnote{Some authors use the terms \emph{synchronizing word} or \emph{directing
word}.%
}. If an automaton has some reset word, we call it a \emph{synchronizing}
\emph{automaton}. For deterministic automata these definitions are
clear, while for more general types of automata (partial, nondeterministic,
probabilistic, weighted,...) there are multiple variants, each of
its own importance. Several fields of mathematics and engineering
deal with such notions. Classical applications (see \citep{VOL1short})
include model-based testing of sequential circuits, robotic manipulation,
symbolic dynamics, and design of noise-resistant systems \citep{DEY1},
but there are important connections also with information theory \citep{TRS1}
and with formal models of biomolecular processes \citep{BON1}. 

Two particular problems concerning synchronization has gained some
publicity: the Road Coloring Problem and the \v{C}ern\'{y} Conjecture.
The first has been solved by Trahtman \citep{TRA6ijm} in 2007 by
proving that the edges of any aperiodic directed multigraph with constant
out-degree can be colored such that a synchronizing deterministic
automaton arises. Motivation for the Road Coloring Problem comes from
symbolic dynamics \citep{ADL1}. On the other hand, the \v{C}ern\'{y}
Conjecture remains open since 1971 \citep{CER2}. It claims that any
$n$-state synchronizing deterministic automaton has a reset word
of length at most $\left(n-1\right)^{2}$. It is known that there
is always a reset word of length at most $\frac{n^{3}-n}{6}$ \citep{PIN2}%
\footnote{An improved bound published by Trahtman \citep{TRA1} in 2011 has
turned out to be proved incorrectly, see \citep{GON1}.%
}.

We focus on two specific synchronization problems that both generalize
the problem represented by the \v{C}ern\'{y} Conjecture. It may be
surprising that in both the cases the corresponding bounds become
exponential:
\begin{itemlist}
\item In a \emph{partial automaton}, each transition is either undefined
or defined unambiguously. A\emph{ careful reset word} of such automaton
maps all the states to one unique state (using only the defined transitions).
The problem is, for given $n$, how long a shortest careful reset
word of an $n$-state partial automaton may be? It is known that in
general it may have strongly exponential length, i.e., $2^{\Omega\left(n\right)}$.
\item Given a deterministic automaton and a subset $S$ of its states, a
\emph{reset word of $S$} maps all the states of $S$ to one unique
state. The problem is, for given $n$, how long a shortest reset word
of a subset of states in an $n$-state deterministic automaton may
be? Again, it is known that in general it may have length $2^{\Omega\left(n\right)}$.
Note that the deterministic automata under consideration do not need
to be synchronizing.
\end{itemlist}
Both these questions, concerning also some variants and restrictions,
have been studied since 1970's but until recently none of the presented
bad cases (i.e., lower bounds) have consisted of automata with two-letter
or any other fixed-size alphabets. The definitions above allow the
alphabet to grow with growing number of states, which offers a very
strong tool for constructing series that witness strongly exponential
lower bounds. 

In 2013 Martyugin \citep{MAR5} gives a lower bound of the form $2^{\Omega\left(\frac{n}{\log n}\right)}$
that applies to both the problems restricted to automata with two-letter
alphabets (i.e., \emph{binary automata}). Moreover, for each fixed
alphabet size $m\ge2$, the author of \citep{MAR5} provides a specific
multiplicative constant in the exponent. There is a simple construction
(Lemma \ref{lem: bin}) guaranteeing that existence of a lower bound
of the form $2^{\Omega\left(n\right)}$ for any particular alphabet
size $m$ implies a lower bound of the same form for any other $m\ge2$
as well. However it remained as an open question if the lower bounds
can be raised to $2^{\Omega\left(n\right)}$ in such a way.

In the present paper we give the answer: We present a lower bound
of the form $2^{\Omega\left(n\right)}$ that applies to both the problems
restricted to binary automata. Moreover, we introduce a technique
for applying the lower bounds even under the restriction to strongly
connected binary automata. The main results are expressed by Theorem
\ref{thm: main} in Section \ref{sec:The-New-Lower Bounds}.

\section{Preliminaries\label{sec:Preliminaries}}

\subsection{Partial finite automata}

A \emph{partial finite automaton }(\emph{PFA})\emph{ }is a triple
$A=\left(Q,X,\delta\right)$, where $Q$ and $X$ are finite sets
and $\delta:Q\times X\rightarrow Q$ is a partial function. Elements
of $Q$ are called \emph{states}, $X$ is the \emph{alphabet}. The
\emph{transition function} $\delta$ can be naturally extended to
$Q\times X^{\star}\rightarrow Q$ by defining
\[
\delta\!\left(s,vx\right)=\delta\!\left(\delta\!\left(s,v\right),x\right)
\]
inductively for $x\in X,v\in X^{\star}$ if the right-hand side is
defined. We extend $\delta$ also with
\[
\delta\!\left(S,w\right)=\left\{ \delta\!\left(s,w\right)\mid s\in S\mbox{, }\delta\!\left(s,w\right)\mbox{ defined}\right\} 
\]
for each $S\subseteq Q$ and $w\in X^{\star}$. A PFA is \emph{deterministic}
(\emph{DFA}) if $\delta$ is a total function. A PFA $\left(Q,X,\delta\right)$
is said to be \emph{strongly connected}%
\footnote{Some authors use the term \emph{transitive} \emph{automaton}.%
}\emph{ }if
\[
\left(\forall r,s\in Q\right)\left(\exists w\in X^{\star}\right)\delta\!\left(r,w\right)=s.
\]
A state $s\in Q$ is a \emph{sink state }if $\delta\!\left(s,x\right)=s$
for each $x\in X$. Clearly, if a nontrivial PFA has a sink state,
it is impossible for the PFA to be strongly connected. The class of
all strongly connected PFA and the class of all PFA with $k$-letter
alphabets are denoted by $\mathcal{SC}$ and $\mathcal{AL}_{k}$ respectively.
Automata from $\mathcal{AL}_{2}$ are called \emph{binary}.

\subsection{Careful synchronization}

For a given PFA, we call $w\in X^{\star}$ a \emph{careful reset word}
if
\[
\left(\exists r\in Q\right)\left(\forall s\in Q\right)\delta\!\left(s,w\right)=r.
\]
If such a word exists, the automaton is \emph{carefully synchronizing}.
There are also notions that describe ,,less careful'' variants of
synchronization of PFA. E.g., both the definition of \emph{D2-directing}
(see Section \ref{sub:PFAvsNFA}) and the definition of a \emph{reset
word }from \citep{BER6} give conditions that may hold for words that
are not careful reset words. However, all the notions are identical
if we consider only DFA. In such cases we can just use terms \emph{reset
word} and \emph{synchronizing}, without the adjective \emph{careful}.
In DFA, each word having a reset word as a factor is also a reset
word. Note that in a carefully synchronizing PFA, there is always
at least one $x\in X$ such that $\delta\!\left(s,x\right)$ is defined
on each $s\in Q$.

We use the following notation consistent with \citep{MAR5}. For a
PFA $A$, let $\car\!\left(A\right)$ denote the length of a shortest
careful reset word of $A$. If there is no such word, we put $\car\!\left(A\right)=0$.
For each $n\ge1$, let $\car\!\left(n\right)$ denote the maximum
value of $\car\!\left(A\right)$ taken over all $n$-state PFA $A$.
It is easy to see that $\car\!\left(n\right)\le2^{n}-n-1$ for each
$n$, but this upper bound has been pushed down to $\car\!\left(n\right)=O\left(n^{2}\cdot4^{\frac{n}{3}}\right)$
by Gazdag, Iv\'an, and Nagy-Gy\"{o}rgy \citep{GAZ1} in 2009.

\subsection{Subset synchronization}

Even if a PFA is not synchronizing, there could be various subsets
$S\subseteq Q$ such that 
\[
\left(\exists r\in Q\right)\left(\forall s\in S\right)\delta\!\left(s,w\right)=r
\]
for some word $w\in X^{\star}$. We say that such $S$ is \emph{carefully
synchronizable} in\emph{ }$A$ and in the opposite case we say it
is \emph{blind }in \emph{$A$}. The word $w$ is called a \emph{careful
reset word of $S$ }in \emph{$A$}. If $A$ is a DFA, we call $w$
just a \emph{reset word} of a \emph{synchronizable} subset $S$ in
$A$. Such words concerning DFA are of our main interest. They lack
some of the elegant properties of classical reset words of DFA (i.e.,
reset words of $S=Q$), particularly a word $w$ having a factor $v$
which is a reset word of $S$ need not to be itself a reset word of
$S$. In fact, if we choose a subset $S$ and a word $w$, it is possible
for the set $\delta\!\left(S,w\right)$ to be blind even if the set
$S$ is synchronizable.

For a PFA $A$, and $S\subseteq Q$ let $\csub\!\left(A\right)$ denote
the minimum length of a careful reset word of $S$ in $A$. If $S$
is blind, we set $\csub\!\left(A,S\right)=0$. If $A$ is a DFA we
write $\sub\!\left(A\right)$ instead of $\csub\!\left(A\right)$.
For each $n\ge1$, let $\csub\!\left(n\right)$ (and $\sub\!\left(n\right)$)
denote the maximum value of $\csub\!\left(A\right)$ taken over all
$n$-state PFA $A$ (or $n$-state DFA $A$ respectively) and all
their subsets of states. It is easy to see that 
\[
\sub\!\left(n\right)\le\csub\!\left(n\right)\le2^{n}-n-1
\]
for each $n$. Our strongly exponential lower bound applies to $\sub\!\left(n\right)$
and thus to $\csub\!\left(n\right)$ as well. The values $\csub\!\left(n\right)$
play only an auxiliary role in the present paper. 

If an automaton $A=\left(Q,X,\delta\right)$ and a subset $S\subseteq Q$
are given (possibly with $S=Q$), we say that $s\in Q$ is \emph{active
after }(or \emph{during}) \emph{the application of $u\in X^{\star}$
}if $s\in\delta\!\left(S,u\right)$ (or $s\in\delta\!\left(S,v\right)$
for a prefix $v$ of $u$, respectively).

\section{Previously Known Lower Bounds\label{sec:Previously-Known-Lower}}

Let $\mathcal{M}$ be a class of automata. For each $n$ let $\mathcal{M}_{\leq n}$
be the class of all automata lying in $\mathcal{M}$ and having at
most $n$ states. For each $n\ge1$ we extend our notation of $\car{}^{\mathcal{M}}\!\left(n\right)$,
$\sub^{\mathcal{M}}\!\left(n\right)$, and $\csub^{\mathcal{M}}\!\left(n\right)$,
denoting the maximum values of $\car\!\left(A\right)$, $\sub\!\left(A\right)$
and $\csub\!\left(A\right)$ taken over $A\in\mathcal{M}_{\le n}$.
In the cases of $\sub^{\mathcal{M}}\!\left(n\right)$ and $\csub^{\mathcal{M}}\!\left(n\right)$,
we use the notion in the obvious way even if $\mathcal{M}$ is a class
of pairs automaton-subset. All such notions we informally call \emph{synchronization
thresholds.} 

In 1976 Burkhard \citep{BUR1} showed that for any $n\ge2$ and $k\le n-2$
it is not hard to produce an $n$-state, $\binom{n-2}{k-1}$-letter
DFA with a $k$-state subset $S$ such that $\sub\!\left(A,S\right)\ge{n-2 \choose k-1}$.
If we set $k=\frac{n}{2}$ and use Stirling's approximation to check
that
\[
\binom{n}{\frac{n}{2}}\approx\sqrt{\frac{2}{\pi}}\cdot\frac{2^{n}}{\sqrt{n}},
\]
we get
\[
\sub\!\left(n\right)=\Omega\!\left(\frac{2^{n}}{\sqrt{n}}\right).
\]
The threshold $\car\!\left(n\right)$ was initially studied in 1982
by Goral\v{c}\'{i}k et al. \citep{GK1}, together with several related
problems. The authors show that for infinitely many $n$ there is
a permutation of $n$ states having order at least $\left(\sqrt[3]{n}\right)!$
and they use it to prove that $\car\!\left(n\right)\geq\left(\sqrt[3]{n}\right)!$.
The construction can be easily (e.g., using our Lemma \ref{lem: car->sub})
modified to establish $\sub\!\left(n\right)\geq\left(\sqrt[3]{n}\right)!$
as well, as it was later re-discovered in the paper \citep{LY1}.
Though exceeded by $\Omega\!\left(\frac{2^{n}}{\sqrt{n}}\right)$,
the later lower bound of $\sub\!\left(n\right)$ remains interesting
since the proof uses binary alphabets only.

In \citep{ITO1short}, Ito and Shikishima-Tsuji proved that $\car\!\left(n\right)\ge2^{\frac{n}{2}}$
and the construction was subsequently improved by Martyugin \citep{MAR6}
in order to reach $\car\!\left(n\right)\geq3^{\frac{n}{3}}$. Again,
the construction can be applied to subsets, so we get $\sub\!\left(n\right)\geq3^{\frac{n}{3}}$.
However, the last proofs seem to use very artificial examples of automata:
\begin{itemlist}
\item In the series, the alphabet size grows linearly with the growing number
of states - the proofs rely on the convention of measuring the \emph{size
}of an automaton only by the number of states. The results say nothing
about the thresholds $\sub{}^{\mathcal{AL}_{k}}\!\left(n\right)$
or $\car{}^{\mathcal{AL}_{k}}\!\left(n\right)$ for any $k\ge2$.
In 2013, Martyugin \citep{MAR5} proves that
\[
\car^{\mathcal{AL}_{2}}\!(n)>3^{\frac{n}{6\cdot\log_{2}n}}
\]
and 
\[
\car^{\mathcal{AL}_{k}}\!(n)>3^{\frac{n}{3\cdot\log_{m\text{\textminus}1}n}}
\]
for each $k\ge3$, which applies in a similar form also to subset
synchronization. However, it remained unclear whether $\car^{\mathcal{AL}_{k}}\!\left(n\right)=2^{\Omega\left(n\right)}$
or $\sub^{\mathcal{AL}_{k}}\!\left(n\right)=2^{\Omega\left(n\right)}$
for some $k\ge2$. Here we confirm this for $k=2$, so for any greater
$k$ the claim follows easily.
\item In the case of subset synchronization, the DFA have sink states, typically
two of them in each automaton. Use of sink states is a very strong
tool for designing automata having given properties, but in practice
such automata seem very special. They represent unstable systems balancing
between different deadlocks. The very opposite are strongly connected\emph{
}automata. Does the threshold remain so high if we consider only strongly
connected DFA? Unfortunately, we show below that it does, even if
we restrict the alphabet size to a constant. We introduce \emph{swap
congruences} as an alternative to sink states. \\
Note that in the case of careful synchronization, any lower bound
of $\car\!\left(n\right)$ applies easily to $\car^{\mathcal{SC}}\!\left(n\right)$
using a simple trick from Lemma \ref{lem: car sc}. Moreover, for
suitable series the alphabet size is increased only by a constant.
\end{itemlist}
In short, in the present paper we prove that
\begin{eqnarray*}
\sub^{\mathcal{AL}_{2}\cap\mathcal{SC}}\!\left(n\right) & = & 2^{\Omega\left(n\right)},\\
\car{}^{\mathcal{AL}_{2}\cap\mathcal{SC}}\!\left(n\right) & = & 2^{\Omega\left(n\right)}.
\end{eqnarray*}
The new bounds are tight in the sense of $\car\!\left(n\right)=2^{\theta\left(n\right)}$
and $\sub\!\left(n\right)=2^{\mathcal{\theta}\left(n\right)}$.

\section{Reductions between Thresholds\label{sec:Reductions-between-Thresholds}}

This section prepares the ground for the results presented in Section
\ref{sec:The-New-Lower Bounds} by introducing basic principles and
relationships concerning the studied thresholds. The principles are
not innovative, except for the method using \emph{swap congruences
}described in the paragraph \ref{sub:Strong-connectivity}, dealing
with strong connectivity in subset synchronization.

As noted above, many of the lower bounds of $\car\!\left(n\right)$
and $\sub\!\left(n\right)$ found in the literature were formulated
for only one of the notions but used ideas applicable to the other
as well. The key method used in the present paper is of this kind
again. However, we are not able to calculate any of the thresholds
from the other exactly, so we at least show several related inequalities
and then use some of them in Section \ref{sec:The-New-Lower Bounds}
to prove the main results. We use the term \emph{reduction} since
we prove the inequalities by transforming an instance of a problem
to an instance of another problem.

\subsection{Determinization by adding sink states }

The following inequality is not a key tool of the present paper; we
prove it in order to illustrate that even careful subset synchronization
is not much harder than subset synchronization itself. Recall that
trivially $\car\!\left(n\right)\le\csub\!\left(n\right)$ and $\sub\!\left(n\right)\le\csub\!\left(n\right)$
for each $n$.
\begin{lemma}
\label{lem: car->sub}For each $n\ge1$ it holds that
\[
\csub\!\left(n\right)\le\sub\!\left(n+2\right)-1.
\]
\end{lemma}
\begin{proof}
Take any PFA $A=\left(Q_{A},X_{A},\delta_{A}\right)$ with a carefully
synchronizable subset $S_{A}\subseteq Q_{A}$ and choose a shortest
careful reset word $w\in X^{\star}$ of $S_{A}$ with $\delta_{A}\!\left(s,w\right)=r_{0}$
for each $s\in S_{A}$. We construct a DFA $B=\left(Q_{B},X_{B},\delta_{B}\right)$
and a synchronizable subset $S_{B}\subseteq Q_{B}$ such that $\sub\!\left(B,S_{B}\right)\ge\left|w\right|+1$.
Let us set
\[
\begin{array}{cc}
\begin{aligned}Q_{B} & =Q_{A}\cup\left\{ \mathrm{D},\overline{\mathrm{D}}\right\} ,\\
X_{B} & =X_{A}\cup\left\{ \omega\right\} ,
\end{aligned}
\hspace{20bp} & \begin{aligned}\delta_{B}\!\left(\mathrm{D},x\right) & =\mathrm{D},\\
\delta_{B}\!\left(\mathrm{\overline{D}},x\right) & =\mathrm{\overline{D}}
\end{aligned}
\end{array}
\]
for each $x\in X_{B}$, and
\[
\begin{array}{cc}
\delta_{B}\!\left(s,x\right)=\begin{cases}
\delta_{A}\!\left(s,x\right) & \mbox{if defined},\\
\overline{\mathrm{D}} & \mbox{otherwise},
\end{cases}\hspace{20bp} & \delta_{B}\!\left(s,\omega\right)=\begin{cases}
\mathrm{D} & \mbox{if }s=r_{0},\\
\overline{\mathrm{D}} & \mbox{otherwise}
\end{cases}\end{array}
\]
for each $s\in Q_{A},x\in X_{A}$. Denote $S_{B}=S_{A}\cup\left\{ \mathrm{D}\right\} .$
The word $w\omega$ witnesses that the subset $S_{B}$ is synchronizable.
On the other hand, let $v$ be any reset word of $S_{B}$. Since $\mathrm{D}$
is a sink state and $\mathrm{D}\in S_{B}$, we have $\delta_{B}\!\left(v,s\right)=\mathrm{D}$
for each $s\in S_{B}$. Thus:
\begin{itemlist}
\item The state $\overline{\mathrm{D}}$ is not active during the application
of $v$.
\item There need to be an occurrence of $\omega$ in $v$. 
\end{itemlist}
Denote $v=v_{0}\omega v_{1}$, where $v_{0}\in X_{A}^{\star}$ and
$v_{1}\in X_{B}^{\star}$. If $\left|\delta_{B}\!\left(S_{B},v_{0}\right)\cap Q_{A}\right|=1$,
we are done since $v_{0}$ maps all the states of $S_{A}$ to a unique
state using only the transitions defined in $A$, so $\left|v\right|\ge\left|w\right|+1$.
Otherwise, there is some $s\in\delta_{B}\!\left(Q_{B},v_{0}\right)\cap Q_{A}$
such that $s\neq r_{0}$, but then $\delta_{B}\!\left(\omega,s\right)=\overline{\mathrm{D}}$,
which is a contradiction.
\end{proof}

\subsection{Strong connectivity\label{sub:Strong-connectivity}}

First, we show an easy reduction concerning careful synchronization
of strongly connected PFA. We use a simple trick: A letter that is
defined only on a single state cannot appear in a shortest careful
reset word, so one can make a PFA strongly connected by adding such
letters. The number of new letters needed may be reduced by adding
special states, but the simple variant described by Lemma \ref{lem: car sc}
is illustrative and strong enough for our purpose.

For each $j\ge0$ we define the class $\mathcal{C}_{j}$ of PFA as
follows. A PFA $A=\left(Q,X,\delta\right)$ belongs to $\mathcal{C}_{j}$
if there are $j$ pairs $\left(r_{1},q_{1}\right),\dots,\left(r_{j},q_{j}\right)\in Q\times Q$
such that adding transitions of the form $r_{i}\longrightarrow q_{i}$
for each $i=1,\dots,j$ makes the automaton strongly connected. Note
that $\mathcal{C}_{0}=\mathcal{SC}$.
\begin{lemma}
\label{lem: car sc}For each $n,k,j\ge1$ it holds that
\[
\car^{\mathcal{AL}_{k}\cap\mathcal{C}_{j}}\!\left(n\right)\le\car^{\mathcal{AL}_{k+j}\cap\mathcal{SC}}\!\left(n\right).
\]
\end{lemma}
\begin{proof}
Take any PFA $A=\left(Q,X_{A},\delta_{A}\right)\in\mathcal{AL}_{k}\cap\mathcal{C}_{j}$
together with the pairs $\left(r_{1},q_{1}\right),\dots,\left(r_{j},q_{j}\right)\in Q\times Q$
from the definition of $\mathcal{C}_{j}$. We construct a PFA $B=\left(Q,X_{B},\delta_{B}\right)$
where $X_{B}=X_{A}\cup\left\{ \psi_{1},\dots,\psi_{j}\right\} $,
$\delta_{B}\!\left(s,x\right)=\delta_{A}\!\left(s,x\right)$ for $x\in X_{A}$
and $s\in Q$, and
\[
\delta_{B}\!\left(s,\psi_{i}\right)=\begin{cases}
q_{i} & \mbox{if }s=r_{i},\\
\und & \mbox{otherwise}
\end{cases}
\]
for $i=1\dots,,j'$ and $s\in Q$. Now it is easy to check that $B$
is strongly connected and that $\car\!\left(B\right)=\car\!\left(A\right)$.
\end{proof}
Second, we present an original method concerning subset synchronization
of strongly connected DFA. All the lower bounds applicable to $\sub\!\left(n\right)$
that we have found in the literature used two sink states (deadlocks)
to force application of particular letters during a synchronization
process. A common step in such proof looks like \emph{,,The letter
$x$ cannot be applied since that would make the sink state $\overline{\mathrm{D}}$
active, while another sink state $\mathrm{D}$ is active all the time''}.
In order to prove a lower bound of $\sub^{\mathcal{SC}}\!\left(n\right)$,
we have to develop an alternative mechanism. Our mechanism relies
on \emph{swap congruences}:

Recall that, given a DFA $A=\left(Q,X,\delta\right)$, an equivalence
relation $\rho\subseteq Q\times Q$ is a \emph{congruence }if 
\[
r\rho s\Rightarrow\delta\!\left(r,x\right)\rho\,\delta\!\left(s,x\right)
\]
for each $x\in X$. We say that a congruence $\rho$ is a \emph{swap
congruence} of a DFA if, for each equivalence class $C$ of $\rho$
and each letter $x\in X$, the restricted function $\delta:C\times\left\{ x\right\} \rightarrow Q$
is injective. The key property of swap congruences is the following.
\begin{lemma}
Let $A=\left(Q,X,\delta\right)$ be a DFA, let $\rho\subseteq Q\times Q$
be a swap congruence and take any $S\subseteq Q$. If there are any
$r,s\in S$ with $r\neq s$ and $r\rho s$, the set $S$ is blind.\end{lemma}
\begin{proof}
Because $r$ and $s$ lie in a common equivalence class of $\rho$,
by the definition of a swap congruence we have $\delta\!\left(r,w\right)\neq\delta\!\left(s,w\right)$
for any $w\in X^{\star}$.
\end{proof}
Thus, the alternative mechanism relies on arguments of the form \emph{,,The
letter $x$ cannot be applied since that would make both the states
$r,s$ active, while it holds that $r\rho s$''}. It turns out that
our results based on the method can be derived from more transparent
but not strongly connected constructions by the following reduction
principle: 
\begin{lemma}
\label{lem: sub sc}For each $n\ge1$ it holds that 
\[
\sub\!\left(n\right)\le\sub^{\mathcal{SC}}\!\left(2n+2\right)-1.
\]
Moreover, for each $n,k\ge1$ and $j\ge2$ it holds that
\[
\sub^{\mathcal{AL}_{k}\cap\mathcal{C}_{j}}\!\left(n\right)\le\sub^{\mathcal{AL}_{k+j}\cap\mathcal{SC}}\!\left(2n+2\right)-1.
\]
\end{lemma}
\begin{proof}
The first claim follows easily from the second one. So, take any DFA
$A=\left(Q_{A},X_{A},\delta_{A}\right)\in\mathcal{AL}_{k}\cap\mathcal{C}_{j}$
together with the pairs $\left(r_{1},q_{1}\right),\dots,\left(r_{j},q_{j}\right)\in Q_{A}\times Q_{A}$
from the definition of $\mathcal{C}_{j}$ and let $S\subseteq Q_{A}$
be synchronizable. We construct a strongly connected DFA $B=\left(Q_{B},X_{B},\delta_{B}\right)$
and a subset $S_{B}\subseteq Q_{B}$ such that $\sub\left(B,S_{B}\right)\ge\sub\left(A,S_{A}\right)+1$.
Let us set
\begin{eqnarray*}
Q_{B} & = & \left\{ s,\overline{s}\mid s\in Q_{A}\right\} \cup\left\{ \mathrm{E},\overline{\mathrm{E}}\right\} ,\\
X_{B} & = & X_{A}\cup\left\{ \psi_{1},\dots,\psi_{j}\right\} .
\end{eqnarray*}
We want the relation
\[
\rho=\left\langle \left(s,\overline{s}\right)\mid s\in Q_{A}\cup\left\{ \mathrm{E}\right\} \right\rangle ,
\]
where $\left\langle \dots\right\rangle $ denotes an equivalence closure,
to be a swap congruence. Regarding this requirement, it is enough
to define $\delta_{B}$ on $Q_{A}\cup\left\{ \mathrm{E}\right\} $.
The remaining transitions are forced by the injectivity on the equivalence
classes. We set
\[
\begin{array}{cc}
\delta_{B}\!\left(s,x\right)=\delta_{A}\!\left(s,x\right),\hspace{20bp} & \delta_{B}\!\left(\mathrm{E},x\right)=\mathrm{E}\end{array}
\]
for any $s\in Q_{A},x\in X_{A}$, while the letters $\psi_{1},\dots,\psi_{j}$
act as follows: 
\[
\begin{array}{cc}
\delta_{B}\!\left(s,\psi_{I}\right)=\begin{cases}
q_{I} & \mbox{if }s=r_{I},\\
\overline{q_{I}} & \mbox{otherwise},
\end{cases}\hspace{20bp} & \delta_{B}\!\left(\mathrm{E},\psi_{I}\right)=q_{I},\\
\delta_{B}\!\left(s,\psi_{i}\right)=\begin{cases}
q_{i} & \mbox{if }s=r_{i},\\
\overline{\mathrm{E}} & \mbox{otherwise},
\end{cases}\hspace{20bp} & \delta_{B}\!\left(\mathrm{E},\psi_{i}\right)=\mathrm{E}\hspace{4bp}
\end{array}
\]
for $s\in Q_{A}$ and $i\neq I$, where $I$ is chosen such that for
a reset word $w$ of $S_{A}$ in $A$ with $\delta_{A}\!\left(s,w\right)=r_{0}$,
the state $r_{I}$ is reachable from $r_{0}$. It is easy to see that
such $I$ exists for any $r_{0}\in S_{A}$. We set $S_{B}=S_{A}\cup\left\{ \mathrm{E}\right\} $.
\begin{itemlist}
\item First, note that the set $S_{B}$ is synchronizable in $B$ by the
word $wu\psi_{I}$ where $u\in X_{A}^{\star}$ such that $\delta_{A}\!\left(r_{0},u\right)=r_{I}$. 
\item On the other hand, let $v$ be a reset word of $S_{B}$ in $B$. The
word $v$ necessarily contains some $\psi_{i}$ for $i\in\left\{ 1,\dots,j\right\} $,
so we can write $v=v_{0}\psi_{i}v_{1}$, where $v_{0}\in X_{A}^{\star},v_{1}\in X_{B}^{\star}$.
If $v_{0}$ is a reset word of $S_{A}$ in $A$, $\left|v\right|\ge\sub\!\left(A,S_{A}\right)+1$
and we are done. Otherwise there is a state $s\neq r_{i}$ in $\delta_{B}\!\left(S,v_{0}\right)$
and we see that both $q_{i}$ and $\overline{q_{i}}$ (if $i=I$)
or both $\mathrm{E}$ and $\overline{\mathrm{E}}$ (if $i\neq I$)
lie in $\delta_{B}\!\left(S,v_{0}\psi_{1}\right)$, which is a contradiction
with properties of the swap congruence $\rho$.
\end{itemlist}
The automaton $B$ is strongly connected since the transitions $r_{i}\overset{\psi_{i}}{\longrightarrow}q_{i}$
and $\overline{r_{i}}\overset{\psi_{i}}{\longrightarrow}\overline{q_{i}}$
for each $i=1,\dots,j$ make both the copies of $A$ strongly connected
and there are transitions $\mathrm{E}\overset{\psi_{I}}{\longrightarrow}q_{I}$,
$s\overset{\psi_{i}}{\longrightarrow}\overline{\mathrm{E}}$ , $\overline{\mathrm{E}}\overset{\psi_{I}}{\longrightarrow}\overline{q_{2}}$,
and $\overline{s}\overset{\psi_{i}}{\longrightarrow}\mathrm{E}$ for
some $i\neq I$ and $s\neq r_{i}$.
\end{proof}

\subsection{A special case of subset synchronization}

We are not aware of any general bad-case reduction from subset synchronization
to careful synchronization. Here we suggest a special class (denoted
by $\mathcal{M}_{\mathrm{P}}$) of pairs automaton-subset such that
the instances from the class are in certain sense reducible to careful
synchronization. The main construction of the present paper (i.e.,
the proof of Lemma \ref{lem: de bruijn automaton}) yields instances
of subset synchronization that fit to this class. We use the following
definitions:
\begin{itemlist}
\item Given a PFA $A=\left(Q,X,\delta\right)$ and a carefully synchronizable
subset $S\subseteq Q$, the $S$\emph{-relevant part} of $A$ is 
\begin{eqnarray*}
Q_{A,S} & = & \bigcup_{w\in W_{S}}\delta\!\left(S,w\right),
\end{eqnarray*}
where $W_{S}$ is the set of prefixes of careful reset words of $S$
in $A$. The $S$-\emph{relevant automaton }of $A$ is $R_{A,S}=\left(Q_{A,S},X,\delta_{A,S}\right)$,
where
\[
\delta_{A,S}\!\left(s,x\right)=\begin{cases}
\delta\!\left(s,x\right) & \mbox{if }\delta\!\left(s,x\right)\in Q_{A,S},\\
\und & \mbox{otherwise}
\end{cases}
\]
for each $s\in Q_{A,S}$ and $x\in X$.
\item The class $\mathcal{M}_{\mathrm{P}}$ is defined as follows. For any
PFA $A=\left(Q,X,\delta\right)$ and any carefully synchronizable
$S\subseteq Q$, the pair $\left\langle A,S\right\rangle $ lies in
$\mathcal{M}_{\mathrm{P}}$ if there are subsets $P_{1},\dots,P_{\left|S\right|}\subseteq Q$
such that:

\begin{itemlist}
\item The sets $P_{1},\dots,P_{\left|S\right|}$ are disjoint and $\bigcup_{i=1}^{\left|S\right|}P_{i}=Q_{A,S}$. 
\item For each $v\in X^{\star}$ such that $\delta\!\left(s,u\right)\in Q_{A,S}$
for any prefix $u$ of $v$ and any $s\in S$, it holds that $v$
is a careful reset word of $S$, or 
\[
\left|\delta\!\left(S,v\right)\cap P_{i}\right|=1
\]
for each $i=1,\dots,\left|S\right|$. In particular, the choice of
empty $v$ implies that 
\[
\left|S\cap P_{i}\right|=1
\]
must hold for each $i=1,\dots,\left|S\right|$.
\end{itemlist}
\item The class $\mathcal{C}_{j}^{\mathrm{R}}$ for $j\ge0$ is defined
as follows. For any PFA $A=\left(Q,X,\delta\right)$ and any carefully
synchronizable $S\subseteq Q$, the pair $\left\langle A,S\right\rangle $
lies in $\mathcal{C}_{j}^{\mathrm{R}}$ if $R_{A,S}\in\mathcal{C}_{j}$.\end{itemlist}
\begin{lemma}
\label{lem:sub->car}For each $n\ge1$ it holds that
\[
\csub^{\mathcal{M}_{\mathrm{P}}}\!\left(n\right)\le\car\!\left(n\right).
\]
Moreover, for each $n,k,j\ge1$ it holds that 
\[
\csub^{\mathcal{AL}_{k}\cap\mathcal{C}_{j}^{\mathrm{R}}\cap\mathcal{M}_{\mathrm{P}}}\!\left(n\right)\le\car^{\mathcal{AL}_{k+1}\cap\mathcal{C}_{j}}\!\left(n\right).
\]
\end{lemma}
\begin{proof}
The first claim follows easily from the second. So, take any $\left\langle A,S\right\rangle \in\mathcal{M}_{\mathrm{P}}$
with $A=\left(Q,X_{A},\delta_{A}\right)$ and $S\subseteq Q$, together
with the sets $P_{1},\dots,P_{\left|S\right|}$ from the definition
of $\mathcal{M}_{\mathrm{P}}$. By adding a letter $\alpha$ to the
automaton $R_{A,S}$, we construct a carefully synchronizing PFA $B=\left(Q_{A,S},X_{B},\delta_{B}\right)$
with $\car\!\left(B\right)\ge\csub\!\left(A,S\right)$. Let $X_{B}=X_{A}\cup\left\{ \alpha\right\} .$
For each $s\in Q_{A,S}$ we find the $i$ such that $s\in P_{i}$
and define 
\[
\delta_{B}\!\left(s,\alpha\right)=q_{i},
\]
where $q_{i}$ is the only state lying in $S\cap P_{i}$, as guaranteed
by the membership in $\mathcal{M}_{\mathrm{P}}$. The letters of $X_{A}$
act in $B$ as they do in $R_{A,S}$.
\begin{itemlist}
\item It is easy to check that the automaton $B$ is carefully synchronizing
by $\alpha w$ for any $w\in X_{A}^{\star}$ that is a careful reset
word of $S$ in $A$.
\item On the other hand, take a shortest careful reset word $v$ of $B$.
If $\alpha$ does not occur in $v$, then $v$ is a careful reset
word of $S$ in $A$, so $\left|v\right|\ge\csub\left(A,S\right)$.
Otherwise, denote $v=v_{0}\alpha v_{1}$ where $v_{0}\in X_{B}^{\star}$
and $v_{1}\in X_{A}^{\star}$. By the membership in $\mathcal{M}_{\mathrm{P}}$
we have $\left|\delta\!\left(S,v_{0}\right)\cap P_{i}\right|=1$ for
each $i=1,\dots,\left|S\right|$ and thus $\delta_{B}\!\left(S,v_{0}\alpha\right)=S$.
It follows that $v_{1}$ is a careful reset word of $S$ in $A$,
so $\left|v\right|\ge\csub\left(A,S\right)$.
\end{itemlist}
\end{proof}

\subsection{Decreasing the alphabet size}

The following method is quite simple and has been already used in
the literature \citep{BER4}. It modifies an automaton in order to
decrease the alphabet size while preserving high synchronization thresholds. 
\begin{lemma}
\label{lem: bin}For each $n,k\ge1$ it holds that\end{lemma}
\begin{romanlist}
\item $\sub^{\mathcal{AL}_{k}}\!\left(n\right)\le\sub^{\mathcal{AL}_{2}}\!\left(k\cdot n\right)$
~and~ $\sub^{\mathcal{AL}_{k}\cap\mathcal{SC}}\!\left(n\right)\le\sub^{\mathcal{AL}_{2}\cap\mathcal{SC}}\!\left(k\cdot n\right),$
\item $\car^{\mathcal{AL}_{k}}\!\left(n\right)\le\car^{\mathcal{AL}_{2}}\!\left(k\cdot n\right)$
~and~ $\car^{\mathcal{AL}_{k}\cap\mathcal{SC}}\!\left(n\right)\le\car^{\mathcal{AL}_{2}\cap\mathcal{SC}}\!\left(k\cdot n\right).$\end{romanlist}
\begin{proof}
Take a PFA $A=\left(Q_{A},X_{A},\delta_{A}\right)$ with $X_{A}=\left\{ a_{0},\dots,a_{m}\right\} $.
We define a PFA $B=\left(Q_{B},X_{B},\delta_{B}\right)$ as follows:
$Q_{B}=Q_{A}\times X_{A}$, $X_{B}=\left\{ \alpha,\beta\right\} $,
and
\begin{eqnarray*}
\delta_{B}\!\left(\left(s,a_{i}\right),\alpha\right) & = & \begin{cases}
\left(\delta_{A}\!\left(s,a_{i}\right),a_{0}\right) & \mbox{if }\delta_{A}\!\left(s,a_{i}\right)\mbox{ is defined},\\
\und & \mbox{otherwise},
\end{cases}\\
\delta_{B}\!\left(\left(s,a_{i}\right),\beta\right) & = & \begin{cases}
\left(s,a_{i+1}\right) & \mbox{if }i<m,\\
\left(s,a_{m}\right) & \mbox{if }i=m
\end{cases}
\end{eqnarray*}
for each $i=0,\dots,m$. The construction of $B$ applies to both
the claims:
\begin{romanlist}
\item Let $A$ be a DFA. We choose a synchronizable $S_{A}\subseteq Q_{A}$
and denote $S_{B}=S_{A}\times\left\{ a_{0}\right\} $. It is not hard
to see that reset words of $S_{B}$ in $B$ are in a one-to-one correspondence
with reset words of $S_{A}$ in $A$. A word $a_{i_{1}}\dots a_{i_{d}}\in X_{A}^{\star}$
corresponds to $\left(\beta^{i_{1}}\alpha\right)\dots\left(\beta^{i_{d}}\alpha\right)\in X_{B}^{\star}$. 
\item Let $A$ be carefully synchronizing. We can suppose that $\delta_{A}\!\left(s,a_{m}\right)$
is defined on each $s\in Q_{A}$ since for a carefully synchronizing
PFA there always exists such letter. For any careful reset word $a_{i_{1}}\dots a_{i_{d}}$
of $A$, the word $\beta^{m}\alpha\left(\beta^{i_{1}}\alpha\right)\dots\left(\beta^{i_{d}}\alpha\right)$
is a careful reset word of $B$. On the other hand, any careful reset
word of $B$ is also a careful reset word of the subset $Q_{A}\times\left\{ a_{0}\right\} \subseteq Q_{B}$,
whose careful reset words are in a one-to-one correspondence with
careful reset words of $A$, like in the previous case.
\end{romanlist}
Since $\delta_{B}\!\left(\left(s,a_{m}\right),\alpha\right)$ is defined
for each $s\in S_{A}$, it is not hard to check that if $A$ is strongly
connected, so is $B$.
\end{proof}

\section{The New Lower Bounds\label{sec:The-New-Lower Bounds}}

\subsection{The key construction}

Let us present the central construction of the present paper. We build
a series of DFA with a constant-size alphabet and a constant structure
of strongly connected components, together with subsets that require
strongly exponential reset words. Moreover, the pairs automaton-subset
are of the special kind represented by $\mathcal{M}_{\mathrm{P}}$,
so a reduction to careful synchronization of PFA, as introduced in
Lemma \ref{lem:sub->car}, is possible.
\begin{lemma}
\label{lem: de bruijn automaton}For infinitely many $m\ge1$ it holds
that
\[
\sub^{\mathcal{AL}_{4}\cap\mathcal{C}_{2}\cap\mathcal{C}_{2}^{\mathrm{R}}\cap\mathcal{M}_{\mathrm{P}}}\!\left(5m+\log m+3\right)\ge2^{m}\cdot\left(\log m+1\right)+1.
\]
\end{lemma}
\begin{proof}
Suppose $m=2^{k}$. For each $t\in0,\dots,m-1$ we denote by $\tau=\bin\left(t\right)$
the standard $k$-digit binary representation of $t$, i.e., a word
from $\left\{ \mathbf{0},\mathbf{1}\right\} ^{k}$. By a classical
result proved in \citep{stm1} there is a \emph{De Bruijn sequence}
$\xi=\xi_{0}\dots\xi_{m-1}$ consisting of letters $\xi_{i}\in\left\{ \mathbf{0},\mathbf{1}\right\} $
such that each word $\tau\in\left\{ \mathbf{0},\mathbf{1}\right\} ^{k}$
appears exactly once as a cyclic factor of $\xi$ (i.e., it is a factor
or begins by a suffix of $\xi$ and continues by a prefix of $\xi$).
Let us fix such $\xi$. By $\pi\!\left(i\right)$ we denote the number
$t$, whose binary representation $\bin\left(t\right)$ starts by
$\xi_{i}$ in $\xi$. Note that $\pi$ is a permutation of $\left\{ 0,\dots,m-1\right\} $.
Set
\begin{eqnarray*}
Q & = & \left(\left\{ 0,\dots,m-1\right\} \times\left\{ \mathbf{0},\mathbf{0^{\downarrow}},\mathbf{1},\mathbf{1^{\downarrow}},\mathbf{1^{\uparrow}}\right\} \right)\cup\left\{ \mathrm{C}_{0},\dots,\mathrm{C}_{k},\mathrm{D},\overline{\mathrm{D}}\right\} ,\\
X & = & \left\{ \mathbf{0},\mathbf{1},\kappa,\omega\right\} ,\\
S & = & \left(\left\{ 0,\dots,m-1\right\} \times\left\{ \mathbf{0}\right\} \right)\cup\left\{ \mathrm{C}_{0},\mathrm{D}\right\} .
\end{eqnarray*}
Figure 
\begin{figure}
\begin{centering}
\includegraphics{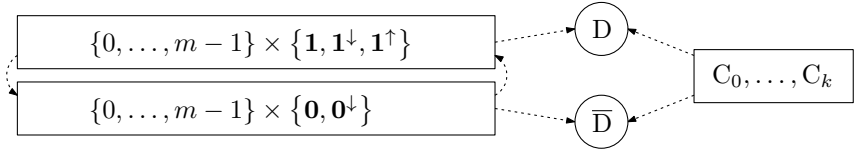}
\par\end{centering}

\caption{A connectivity pattern of the automaton $A$.\label{fig: overall}}
\end{figure}
\ref{fig: overall} visually distinguishes main parts of $A$. The
states $\mathrm{D}$ and $\overline{\mathrm{D}}$ are sink states.
Together with $\mathrm{D}\in S$ it implies that any reset word of
$S$ takes the states of $S$ to $\mathrm{D}$ and that the state
$\overline{\mathrm{D}}$ must not become active during its application.
The states $\mathrm{C}_{0},\dots,\mathrm{C}_{k}$ guarantee that any
reset word of $S$ lies in
\begin{equation}
\left(\left\{ \mathbf{0,1}\right\} ^{k}\kappa\right)^{\star}\omega X^{\star}.\label{eq: rlang}
\end{equation}
Indeed, as defined by Figure \ref{fig: watch}, no other word takes
$\mathrm{C}_{0}$ to $\mathrm{D}$. Let the letter $\omega$ act as
follows:
\begin{eqnarray*}
\left\{ 0,\dots,m-1\right\} \times\left\{ \mathbf{1}\right\} ,\mathrm{C}_{0},\mathrm{D} & \;\;\overset{\omega}{\longrightarrow}\;\; & \mathrm{D},\;\;\;\;\;\;\;\;\;\;\;\;\\
\left\{ 0,\dots,m-1\right\} \times\left\{ \mathbf{0},\mathbf{0^{\downarrow}},\mathbf{1^{\downarrow}},\mathbf{1^{\uparrow}}\right\} ,\mathrm{C}_{1},\dots,\mathrm{C}_{\log m},\overline{\mathrm{D}} & \;\;\overset{\omega}{\longrightarrow}\;\; & \overline{\mathrm{D}}.\;\;\;\;\;\;\;\;\;\;\;\;
\end{eqnarray*}
We see that $\omega$ maps each state to $\mathrm{D}$ or $\overline{\mathrm{D}}$.
This implies that once $\omega$ occurs in a reset word of $S$, it
must complete the synchronization. In order to map $\mathrm{C}_{0}$
to $\mathrm{D}$, the letter $\omega$ \emph{must }occur, so any shortest
reset word of $S$ is exactly of the form
\begin{equation}
w=\left(\tau_{1}\kappa\right)\dots\left(\tau_{d}\kappa\right)\omega,\label{eq: rlang2}
\end{equation}
where $\tau_{j}\in\left\{ \mathbf{0},\mathbf{1}\right\} ^{k}$ for
each $j$.

\begin{figure}
\begin{centering}
\includegraphics{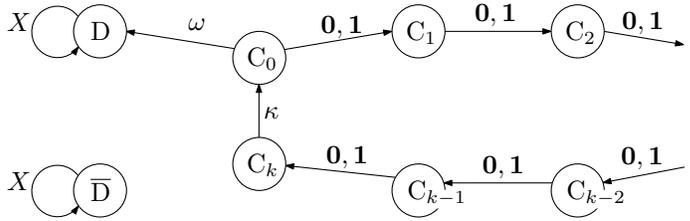}
\par\end{centering}

\caption{\label{fig: watch}A part of $A$. All the outgoing transitions that
are not depicted lead to $\overline{\mathrm{D}}$.}
\end{figure}

The two biggest parts depicted by Figure \ref{fig: overall} are very
similar to each other. The letters $\mathbf{0}$ and $\mathbf{1}$
act on them as follows:
\[
\begin{aligned}\left(i,\mathbf{0}\right)\overset{\mathbf{0}}{\longrightarrow} & \begin{cases}
\left(i+1,\mathbf{0}\right) & \mbox{if }\ensuremath{\xi_{i}=\mathbf{0}},\\
\left(i+1,\mathbf{0^{\downarrow}}\right) & \mbox{if }\ensuremath{\xi_{i}=\mathbf{1}},
\end{cases}\\
\left(i,\mathbf{0}\right)\overset{\mathbf{1}}{\longrightarrow} & \begin{cases}
\overline{\mathrm{D}} & \mbox{if }\ensuremath{\xi_{i}=\mathbf{0}},\\
\left(i+1,\mathbf{0}\right) & \mbox{if }\ensuremath{\xi_{i}=\mathbf{1}},
\end{cases}
\end{aligned}
\mbox{ }\qquad\begin{aligned}\left(i,\mathbf{1}\right)\overset{\mathbf{0}}{\longrightarrow} & \begin{cases}
\left(i+1,\mathbf{1}\right) & \mbox{if }\ensuremath{\xi_{i}=\mathbf{0}},\\
\left(i+1,\mathbf{1^{\downarrow}}\right) & \mbox{if }\ensuremath{\xi_{i}=\mathbf{1}},
\end{cases}\\
\left(i,\mathbf{1}\right)\overset{\mathbf{1}}{\longrightarrow} & \begin{cases}
\left(i+1,\mathbf{1^{\uparrow}}\right) & \mbox{if }\ensuremath{\xi_{i}=\mathbf{0}},\\
\left(i+1,\mathbf{1}\right) & \mbox{if }\ensuremath{\xi_{i}=\mathbf{1}},
\end{cases}
\end{aligned}
\]
and $\left(i,\mathbf{b}\right)\overset{\mathbf{0},\mathbf{1}}{\longrightarrow}\left(i+1,\mathbf{b}\right)$
for each $\mathbf{b}=\mathbf{0^{\downarrow}},\mathbf{1^{\downarrow}},\mathbf{1^{\uparrow}}$,
using the addition modulo $m$ everywhere. For example, Figure \ref{fig: ex dec}
depicts a part of $A$ for $m=8$ and for a particular De Bruijn sequence
$\xi$. Figure \ref{fig: kappa} defines the action of $\kappa$ on
the states $\left\{ i\right\} \times\left\{ \mathbf{0},\mathbf{0^{\downarrow}},\mathbf{1},\mathbf{1^{\downarrow}},\mathbf{1^{\uparrow}}\right\} $
for any $i$, so the automaton $A$ is completely defined.
\begin{figure}
\begin{centering}
\includegraphics{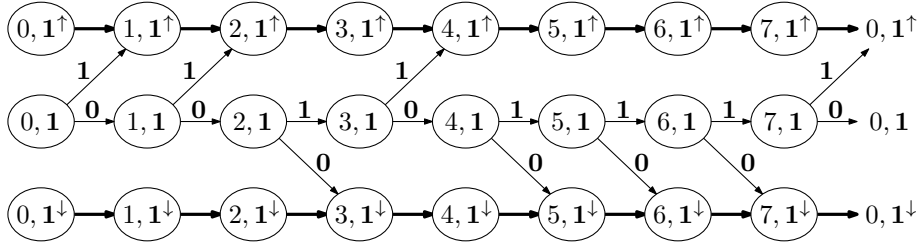}
\par\end{centering}

\caption{\label{fig: ex dec}A part of $A$ assuming $m=8$ and $\xi=\mathbf{00101110}$.
Bold arrows represent both $\mathbf{0},\mathbf{1}$.}
\end{figure}
 
\begin{figure}
\begin{minipage}[t]{0.45\columnwidth}%
\begin{center}
\includegraphics{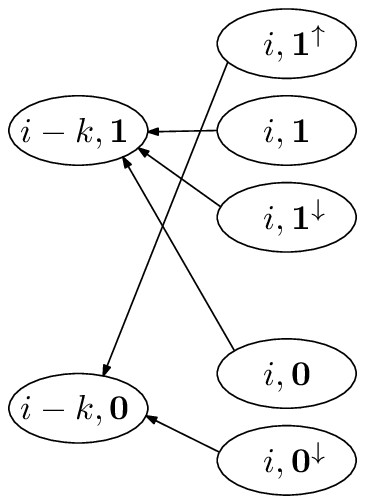}\caption{\label{fig: kappa}The action of the letter $\kappa$, with subtraction
modulo $m$.}

\par\end{center}%
\end{minipage}\hfill{}%
\begin{minipage}[t]{0.45\columnwidth}%
\begin{center}
\includegraphics{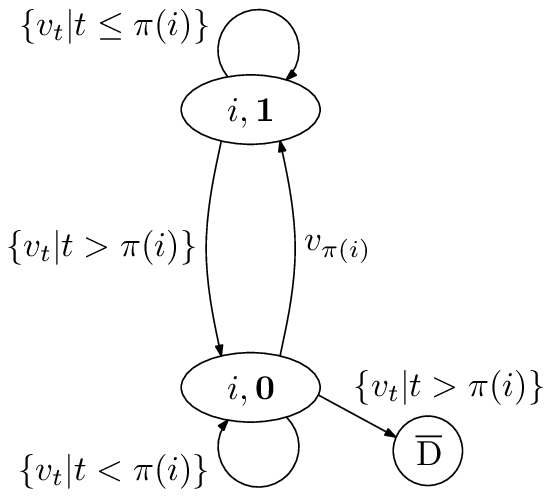}\caption{\label{fig: phase}The action of $v_{0},\dots,v_{m-1}$ on the $i$-th
switch.}

\par\end{center}%
\end{minipage}
\end{figure}

Let $w$ be a shortest reset word of $S$ in $A$. It is necessarily
of the form (\ref{eq: rlang2}), so it makes sense to denote $v_{t}=\bin\left(t\right)\kappa$
and treat $w$ as
\begin{equation}
w=v_{t_{1}}\dots v_{t_{d}}\omega\in\left\{ v_{0},\dots,v_{m-1},\omega\right\} ^{\star}.\label{eq: rlang3}
\end{equation}
The action of each $v_{t}$ is depicted by Figure \ref{fig: phase}.
It is a key step of the proof to confirm that Figure \ref{fig: phase}
is correct. Indeed:
\begin{itemlist}
\item Starting from a state $\left(i,\mathbf{1}\right)$, a word $\bin\left(t\right)$
takes us through a kind of decision tree to one of the states $\left(i+k,\mathbf{1^{\downarrow}}\right),\left(i+k,\mathbf{1}\right),\left(i+k,\mathbf{1^{\uparrow}}\right)$,
depending on whether $t$ is less than, equal to, or greater than
$\pi\!\left(i\right)$, respectively. This is guaranteed by wiring
the sequence $\xi$ into the transition function, see Figure \ref{fig: ex dec}.
The letter $\kappa$ then takes us back to $\left\{ i\right\} \times\left\{ \dots\right\} $,
namely to $\left(i,\mathbf{0}\right)$ or $\left(i,\mathbf{1}\right)$.
\item Starting from a state $\left(i,\mathbf{0}\right)$, we proceed similarly,
but in the case of $t>\pi\!\left(i\right)$ we fall into $\overline{\mathrm{D}}$
during the application of $\bin\left(t\right)$. 
\end{itemlist}
It follows that after applying any prefix $v_{t_{1}}\dots v_{t_{j}}$
of $w$, exactly one of the states $\left(i,\mathbf{0}\right),\left(i,\mathbf{1}\right)$
is active for each $i$. We say that \emph{the $i$-th switch is set
to }\textbf{\emph{$\mathbf{0}$}} \emph{or $\mathbf{1}$} \emph{at
time $j$}. Note that $Q_{A}\backslash\left\{ \overline{\mathrm{D}}\right\} $
is the $S$-relevant part of $A$ and that the sets $\left\{ i\right\} \times\left\{ \mathbf{0},\mathbf{0^{\downarrow}},\mathbf{1},\mathbf{1^{\downarrow}},\mathbf{1^{\uparrow}}\right\} $
for $i=0,\dots,m-1$, together with the sets $\left\{ \mathrm{D}\right\} $
and $\left\{ \mathrm{C}_{0},\dots,\mathrm{C}_{k}\right\} $, can play
the role of $P_{1},\dots,P_{m+2}$ in the definition of $\mathcal{M}_{\mathrm{P}}$.

Observe that at time $d$ all the switches are necessarily set to
$\mathbf{1}$ because otherwise the state $\overline{\mathrm{D}}$
would become active by the application of $\omega$. On the other
hand, at time $0$ all the switches are set to $\mathbf{0}$. We are
going to show that in fact during the synchronization of $S$ the
switches together perform a binary counting from $0$ (all the switches
set to $\mathbf{0}$) to $2^{m}-1$ (all the switches set to $\mathbf{1}$).
For each $i$ the\emph{ significance} of the $i$-th switch is given
by the value $\pi\!\left(i\right)$. So the $\pi^{-1}\!\left(m-1\right)$-th
switch carries the most significant digit, the $\pi^{-1}\!\left(0\right)$-th
switch carries the least significant digit and so on. The number represented
in this manner by the switches at time $j$ is denoted by $\mathfrak{b}_{j}\in\left\{ 0,\dots,2^{m}-1\right\} $.
We claim that $\mathfrak{b}_{j}=j$ for each $j$. Indeed:
\begin{itemlist}
\item At time $0$, all the switches are set to $\mathbf{0}$, we have $\mathfrak{b}_{0}=0$.
\item Suppose that $\mathfrak{b}_{j'}=j'$ for each $j'\leq j-1$. We denote
\begin{equation}
\overline{t_{j}}=\min\left\{ \pi\!\left(i\right)\mid i\mbox{-th switch is set to }\mathbf{0}\mbox{ at time }j-1\right\} \label{eq: t is min pos}
\end{equation}
and claim that $t_{j}=\overline{t_{j}}$. Note that $\overline{t_{j}}$
is defined to be the least significance level at which there occurs
a $\mathbf{0}$ in the binary representation of $\mathfrak{b}_{j-1}$.
Suppose for a contradiction that $t_{j}>\overline{t_{j}}$. By the
definition of $\overline{t_{j}}$ the state $\left(\pi^{-1}\!\left(\overline{t_{j}}\right),\mathbf{0}\right)$
lies in $\delta\!\left(S,v_{t_{1}}\dots v_{t_{j-1}}\right)$. But
$v_{t_{j}}$ takes this state to $\overline{\mathrm{D}}$, which is
a contradiction. Now suppose that $t_{j}<\overline{t_{j}}$. In such
case the application of $v_{t_{j}}$ does not turn any switch from
$\mathbf{0}$ to $\mathbf{1}$, so $\mathfrak{b}_{j}\leq\mathfrak{b}_{j-1}$
and thus at time $j$ the configuration of switches is the same at
it was at time $\mathfrak{b}_{j}$. This contradicts the assumption
that $w$ is a shortest reset word. We have proved that $t_{j}=\overline{t_{j}}$
and it remains only to show that the application of $v_{t_{j}}$ performs
the addition of $1$ and so makes the switches represent the value
$\mathfrak{b}_{j-1}+1$.

\begin{itemlist}
\item Consider an $i$-th switch with $\pi\!\left(i\right)<t_{j}$. By the
definition of $\overline{t_{j}}$, it is set to $\mathbf{1}$ at time
$j-1$ and the word $v_{t_{j}}$ sets it to $\mathbf{0}$ at time
$j$. This is what we need because such switches represent a continuous
leading segment of $\mathbf{1}$s in the binary representation of
$\mathfrak{b}_{j-1}$. 
\item The $\pi^{-1}\!\left(t_{j}\right)$-th switch is set from $\mathbf{0}$
to $\mathbf{1}$ by the word $v_{t_{j}}$.
\item Consider an $i$-th switch with $\pi\!\left(i\right)>t_{j}$. The
switch represents a digit of $\mathfrak{b}_{j-1}$ which is more significant
than the $\overline{t_{j}}$-th digit. As we expect, the word $v_{t_{j}}$
leaves such switch unchanged.
\end{itemlist}
\end{itemlist}
Because $\mathfrak{b}_{d}=2^{m}$, we deduce that $d=2^{m}$ and thus
$\left|w\right|=2^{m}\cdot\left(k+1\right)+1$, assuming that a shortest
reset word $w$ exists. But in fact we have also shown that there
is only one possibility for such $w$ and that it is a true reset
word for $S$. The unique $w$ is of the form (\ref{eq: rlang3}),
where $t_{j}$ is the position of the least significant \textbf{$\mathbf{0}$
}in the binary representation of $j-1$.

The automaton $A$ lies in $\mathcal{C}_{2}\cap\mathcal{C}_{2}^{\mathrm{R}}$
since the addition of $\mathrm{D}\longrightarrow\mathrm{C}_{0}$ and
$\overline{\mathrm{D}}\longrightarrow\left(0,\mathbf{0}\right)$ makes
$A$ strongly connected, while the addition of $\mathrm{D}\longrightarrow\mathrm{C}_{0}$
and $\mathrm{C}_{0}\longrightarrow\left(0,\mathbf{0}\right)$ makes
$R_{A,S}$ strongly connected.
\end{proof}

\subsection{The results}

The following theorem presents the main results of the present paper:
\begin{theorem}
\label{thm: main}For infinitely many $n\ge1$ it holds that\end{theorem}
\begin{romanlist}
\item $\sub^{\mathcal{AL}_{2}\cap SC}\!\left(n\right)\ge2^{\frac{n}{61}},$
\item $\car^{\mathcal{AL}_{2}\cap\mathcal{SC}}\!\left(n\right)\ge2^{\frac{n}{36}}.$\end{romanlist}
\begin{proof}
Lemma \ref{lem: de bruijn automaton} says that
\begin{eqnarray}
2^{m}\cdot\left(\log m+1\right)+1 & \le & \sub^{\mathcal{AL}_{4}\cap\mathcal{C}_{2}\cap\mathcal{C}_{2}^{\mathrm{R}}\cap\mathcal{M}_{\mathrm{P}}}\!\left(5m+\log m+3\right)\label{eq: deBru rev}
\end{eqnarray}
for infinitely many $m\ge1$. Now we apply some of the lemmas from
Section \ref{sec:Reductions-between-Thresholds}:
\begin{romanlist}
\item Lemma \ref{lem: sub sc} extends (\ref{eq: deBru rev}) with
\[
\sub^{\mathcal{AL}_{4}\cap\mathcal{C}_{2}}\!\left(5m+\log m+3\right)\le\sub^{\mathcal{AL}_{6}\cap\mathcal{SC}}\!\left(10m+2\cdot\log m+8\right)-1
\]
and Lemma \ref{lem: bin} adds
\[
\sub^{\mathcal{AL}_{6}\cap\mathcal{SC}}\!\left(10m+2\cdot\log m+8\right)-1\le\sub^{\mathcal{AL}_{2}\cap SC}\!\left(60m+12\cdot\log m+48\right)-1.
\]
We chain the three inequalities and deduce
\begin{eqnarray*}
\sub^{\mathcal{AL}_{2}\cap SC}\!\left(60m+12\cdot\log m+48\right) & \ge & 2^{m}\cdot\left(\log m+1\right)+2,\\
\sub^{\mathcal{AL}_{2}\cap SC}\!\left(61m\right) & \ge & 2^{m},\\
\sub^{\mathcal{AL}_{2}\cap SC}\!\left(n\right) & \ge & 2^{\frac{n}{61}}.
\end{eqnarray*}

\item Lemma \ref{lem:sub->car} extends (\ref{eq: deBru rev}) with
\[
\csub^{\mathcal{AL}_{4}\cap\mathcal{C}_{2}^{\mathrm{R}}\cap\mathcal{M}_{\mathrm{P}}}\!\left(5m+\log m+3\right)\le\car^{\mathcal{AL}_{5}\cap\mathcal{C}_{2}}\!\left(5m+\log m+3\right),
\]
while Lemma \ref{lem: car sc} adds 
\[
\car^{\mathcal{AL}_{5}\cap\mathcal{C}_{2}}\!\left(5m+\log m+3\right)\le\car^{\mathcal{AL}_{7}\cap\mathcal{SC}}\!\left(5m+\log m+3\right)
\]
and Lemma \ref{lem: bin} adds
\[
\car^{\mathcal{AL}_{7}\cap\mathcal{SC}}\!\left(5m+\log m+3\right)\le\car^{\mathcal{AL}_{2}\cap SC}\!\left(35m+7\cdot\log m+21\right).
\]
We chain the four inequalities and deduce:
\begin{eqnarray*}
\car^{\mathcal{AL}_{2}\cap SC}\!\left(35m+7\cdot\log m+21\right) & \ge & 2^{m}\cdot\left(\log m+1\right)+1,\\
\car^{\mathcal{AL}_{2}\cap SC}\!\left(36m\right) & \ge & 2^{m},\\
\car^{\mathcal{AL}_{2}\cap SC}\!\left(n\right) & \ge & 2^{\frac{n}{36}}.
\end{eqnarray*}

\end{romanlist}
\end{proof}
Note that there are more subtle results for less restricted classes
of automata:
\begin{proposition}
It holds that $\sub^{\mathcal{AL}_{2}}\!\left(n\right)\ge2^{\frac{n}{21}}$,
$\car^{\mathcal{AL}_{2}}\!\left(n\right)\ge2^{\frac{n}{26}}$, $\sub^{SC}\!\left(n\right)\ge3^{\frac{n}{6}}$,
and $\car^{\mathcal{SC}}\!\left(n\right)\ge3^{\frac{n}{3}}$ for infinitely
many $n\ge1$.\end{proposition}
\begin{proof}
The first claim follows easily from Lemmas \ref{lem: de bruijn automaton}
and \ref{lem: bin}, the second one requires also using Lemma \ref{lem:sub->car}
first. The third and the last claim follow from applying Lemmas \ref{lem: car->sub}
and \ref{lem: sub sc} (or Lemma \ref{lem: car sc} respectively)
to the construction from \citep{MAR6}.
\end{proof}

\section{Consequences and Reformulations\label{sec:Consequences-and-reformulations}}

\subsection{Computational problems}

It is well known that the decision about synchronizability of a given
DFA is a polynomial time task, even if we also require an explicit
reset word on the output. A lot of work has been done on such algorithms
in effort to make them produce short reset words in short running
time. However, it has been proven that it is both NP-hard and coNP-hard
(it is actually DP-complete) to recognize the length of \emph{shortest}
reset words for a given DFA, while it is still NP-hard to recognize
its upper bounds or to approximate it with a constant factor, see
references in \citep{OLS1} and \citep{BER4}.

On the other hand, there has not been done much research in computational
complexity of problems concerning subset synchronization and careful
synchronization, although they does not seem to have less chance to
emerge in practice. Namely, the first natural problems in these directions
are 

\begin{flushleft}
\begin{tabular}{ll}
\multicolumn{2}{l}{\noun{Subset synchronizability}}\tabularnewline
Input: & $n$-state DFA $A=\left(Q,X,\delta\right)$, $S\subseteq Q$\tabularnewline
Output: & is there some $w\in X^{\star}$ such that $\left|\delta\!\left(S,w\right)\right|=1$?\tabularnewline
\end{tabular}
\par\end{flushleft}

\begin{flushleft}
\begin{tabular}{ll}
\multicolumn{2}{l}{\noun{Careful synchronizability}}\tabularnewline
Input: & $n$-state PFA $A=\left(Q,X,\delta\right)$,\tabularnewline
Output: & is there some $w\in X^{\star}$ such that $\left(\exists r\in Q\right)\left(\forall s\in Q\right)\delta\!\left(s,w\right)=r$?\tabularnewline
\end{tabular}
\par\end{flushleft}

Both these problems, in contrast to the synchronizability of DFA,
are known to be PSPACE-complete:
\begin{theorem}
\emph{\cite{NAT1,SAN1short}}\noun{~Subset synchronizability} is
a PSPACE-complete problem.
\end{theorem}

\begin{theorem}
\emph{\cite{MAR4}}\noun{~Careful synchronizability} is a PSPACE-complete
problem.
\end{theorem}
Note that such hardness is not a consequence of any lower bound of
synchronization thresholds, because an algorithm does not need to
produce an explicit reset word. The proofs of both the theorems above
make use of a result of Kozen \citep{KOZ1}, which establishes that
it is PSPACE-complete to decide if given finite acceptors with a common
alphabet accept a common word. This problem is polynomially reduced
to our problems\noun{ }using the idea of two sink states. Is it possible
to avoid the non-connectivity here? 

In the case of \noun{Careful synchronizability}, the simple trick
from Lemma \ref{lem: car sc} easily reduces the general problem to
the variant restricted to strongly connected automata, and it turns
out that the method of swap congruences is general enough to perform
such reduction also in the case of \noun{Subset synchronizability:}
\begin{theorem}
The following problems are PSPACE-complete:\end{theorem}
\begin{romanlist}
\item \noun{Subset synchronizability }restricted to binary strongly connected
DFA 
\item \noun{Careful synchronizability }restricted to binary strongly connected
PFA \end{romanlist}
\begin{proof}
There are polynomial reductions from the general problems \noun{Subset
synchronizability} and \noun{Careful synchronizability}: Perform the
construction from Lemma \ref{lem: sub sc} (or Lemma \ref{lem: car sc}
respectively) and then the one from Lemma \ref{lem: bin}.
\end{proof}

\subsection{Synchronization thresholds of NFA\label{sub:PFAvsNFA}}

In 1999, Imreh and Steinby \citep{IMR2} introduced three different
synchronization thresholds concerning general non-deterministic finite
automata (NFA). We define an NFA as a pair $A=\left(Q,X,\delta\right)$
where $Q$ is a finite set of states, $X$ is a finite alphabet and
$\delta:Q\times X\rightarrow2^{Q}$ is a total function, extended
in the canonical way to $\delta:Q\times X^{\star}\rightarrow2^{Q}$.
For any $S\subseteq Q$ we denote $\delta\!\left(S,w\right)=\bigcup_{s\in S}\delta\!\left(s,w\right)$.

The key definitions are the following. For an NFA $A$, a word $w\in X^{\star}$
is:
\begin{itemlist}
\item \emph{D1-directing} if there is $r\in Q$ such that $\delta_{A}\!\left(s,w\right)=\left\{ r\right\} $
for each $s\in Q$,
\item \emph{D2-directing }if $\delta_{A}\!\left(s_{1},w\right)=\delta_{A}\!\left(s_{2},w\right)$
for each $s_{1},s_{2}\in Q$,
\item \emph{D3-directing }if there is $r\in Q$ such that $r\in\delta_{A}\!\left(s,w\right)$
for each $s\in Q$.
\end{itemlist}
By $\mathrm{d_{1}}\!\left(A\right),\mathrm{d_{2}}\!\left(A\right),\mathrm{d_{3}}\!\left(A\right)$
we denote the length of shortest D1-, D2-, and D3-directing words
for $A$, or $0$ if there is no such word. By $\mathrm{d_{1}}\!\left(n\right),\mathrm{d_{2}}\!\left(n\right),\mathrm{d_{3}}\!\left(n\right)$
we denote the maximum values of $\mathrm{d_{1}}\!\left(A\right),\mathrm{d_{2}}\!\left(A\right),\mathrm{d_{3}}\!\left(A\right)$
taken over all NFA $A$ with at most $n$ states. Possible restrictions
are marked by superscripts as usual.

It is clear that PFA are a special kind of NFA. Any careful reset
word of a PFA $A$ is D1-, D2-, and D3-directing. On the other hand,
any D1- or D3-directing word of a PFA is a careful reset word. Thus,
we get
\[
\mathrm{d_{2}^{PFA}}\!\left(n\right)\le\mathrm{d_{1}^{PFA}}\!\left(n\right)=\mathrm{d_{3}^{PFA}}\!\left(n\right)=\car\!\left(n\right)
\]
and
\begin{equation}
\begin{array}{cc}
\mathrm{d_{1}}\!\left(n\right)\ge\car\!\left(n\right),\hspace{20bp} & \mathrm{d_{3}}\!\left(n\right)\ge\car\!\left(n\right).\end{array}\label{eq:car-and-d13}
\end{equation}
Note that a D2-directing word $w$ of a PFA $A$ is either a careful
reset word of $A$ or satisfies that $\delta\!\left(\left\{ s\right\} ,w\right)=\emptyset$
for each $s\in Q$. PFA are of a special importance for the threshold
$\mathrm{d_{3}}\!\left(n\right)$ since due to a key lemma from \citep{ITO1short},
for any $n$-state NFA $A$ there is a $n$-state PFA $B$ such that
$\mathrm{d_{3}}\!\left(B\right)\ge\mathrm{d_{3}}\!\left(A\right)$,
so we have $\mathrm{d_{3}}\!\left(n\right)=\mathrm{d_{3}^{PFA}}\!\left(n\right)$
for each $n$.

It is known that $\mathrm{d_{1}}\!\left(n\right)=\Omega\!\left(2^{n}\right)$
\citep{ITO1short} and $\mathrm{d_{3}}\!\left(n\right)=\Omega\!\left(3^{\frac{n}{3}}\right)$
\citep{MAR6} (for upper bounds and further details see \citep{GAZ1}).
Due to the easy relationship similar to (\ref{eq:car-and-d13}), our
strongly exponential lower bounds apply directly to the thresholds
$\mathrm{d_{1}}\!\left(n\right)$ and $\mathrm{d_{3}}\!\left(n\right)$
with the restriction to binary strongly connected NFA:
\[
\begin{array}{cc}
\mathrm{d_{1}^{\mathcal{AL}_{2}\cap\mathcal{SC}}}\!\left(n\right)=2^{\Omega\left(n\right)},\hspace{20bp} & \mathrm{d_{3}^{\mathcal{AL}_{2}\cap\mathcal{SC}}}\!\left(n\right)=2^{\Omega\left(n\right)}.\end{array}
\]

\subsection{Compositional depths}

It has been pointed out by Arto Salomaa \citep{SAL4} in 2001 that
very little is known about the minimum length of a composition needed
to generate a function by a given set of generators. To be more precise,
let us adopt and slightly extend the notation used in \citep{SAL4}.
We denote by $\mathcal{T}_{n}$ the semigroup of all functions from
$\left\{ 1,\dots,n\right\} $ to itself. Given $\mathbf{G}\subseteq\mathcal{T}_{n}$,
we denote by $\left\langle \mathbf{G}\right\rangle $ the subsemigroup
generated by $\mathbf{G}$. Given $\mathbf{F}\subseteq\mathcal{T}_{n}$,
we denote by $\Df\left(\mathbf{G},\mathbf{F}\right)$ the length $k$
of a shortest sequence $g_{1},\dots,g_{k}$ of functions from $\mathbf{G}$
such that $g_{1}\circ\dots\circ g_{k}\in\mathbf{F}$. Finally, denote
\begin{equation}
\D_{n}=\max_{\overline{n}\leq n}\max_{\substack{\mathbf{\mathbf{F},\mathbf{G}\subseteq\mathcal{T}_{\overline{n}}}\\
\mathbf{F}\cap\left\langle \mathbf{G}\right\rangle \neq\emptyset
}
}\Df\left(\mathbf{G},\mathbf{F}\right).\label{eq: def Dn}
\end{equation}
Note that in Group Theory, thresholds like $\D_{n}$ are studied in
the scope of permutations, see \citep{HEL1}. 

From basic connections between automata and transformation semigroups
it follows that various synchronization thresholds can be defined
alternatively by putting additional restrictions to the space of considered
sets $\mathbf{G}$ and $\mathbf{F}$ in the definition (\ref{eq: def Dn})
of the threshold $\D_{n}$:
\begin{romanlist}
\item For the basic synchronization threshold of DFA (may be denoted by
$\car^{\mathrm{DFA}}\!\left(n\right)$), we restrict $\mathbf{F}$
to be exactly the set of $\overline{n}$-ary constant functions. Recall
that a set $\mathbf{G}\subseteq\mathcal{T}_{\overline{n}}$ corresponds
to a DFA $A=\left(\left\{ 1,\dots\overline{n}\right\} ,X,\delta\right)$:
Each $g\in\mathbf{G}$ just encodes the action of certain $x\in X$.
Finding a reset word of $A$ then equals composing transitions from
$\mathbf{G}$ in order to get a constant.
\item For the threshold $\sub\!\left(n\right)$, we restrict $\mathbf{F}$
to be some of the sets 
\[
\mathbf{F}_{S}=\left\{ f\in\mathcal{T}_{\overline{n}}\mid\left(\forall r,s\in S\right)f\left(r\right)=f\left(s\right)\right\} 
\]
 for $S\subseteq\left\{ 1,\dots,\overline{n}\right\} $. Therefore
it holds that $\D_{n}\geq\sub\!\left(n\right).$
\item For $\car\!\left(n\right)$, we should consider an alternative formalism
for PFA, where the ,,undefined'' transitions lead to a special error
sink state. Let the largest number stand for the error state. A careful
reset word should map all the states except for the error state to
one particular non-error state. So, here we restrict 
\begin{eqnarray*}
\mathbf{F} & = & \left\{ f\in\mathcal{T}_{\overline{n}}\mid\left(\forall r,s\in\left\{ 1,\dots,\overline{n}-1\right\} \right)f\left(r\right)=f\left(s\right)\neq\overline{n}\right\} ,\\
\mathbf{G} & \subseteq & \left\{ g\in\mathcal{T}_{n}\mid g\!\left(n\right)=n\right\} .
\end{eqnarray*}
However, in the canonical formalism such $\mathbf{G}\subseteq\mathcal{T}_{n}$
corresponds to a $\left(n-1\right)$-state PFA, so we get $\D_{n}\geq\car\!\left(n-1\right).$
Allowing suitable sets $\mathbf{F}_{S}$ for $S\subseteq\left\{ 1,\dots,\overline{n}-1\right\} $,
we get $\D_{n}\geq\csub\!\left(n-1\right)$ as well.
\end{romanlist}
Arto Salomaa refers to a single nontrivial bound of $\D_{n}$, namely
$\D_{n}\geq\left(\sqrt[3]{n}\right)!$. In fact, he omits a construction
of Kozen \citep[Theorem 3.2.7]{KOZ1} from 1977, which deals with
lengths of \emph{proofs} rather than compositions but witnesses easily
that $\D_{n}=2^{\Omega\left(\frac{n}{\log n}\right)}$.\textbf{ }However,
the lower bound of $\car\!\left(n\right)$ from \citep{ITO1short}
revealed soon that\textbf{ $\D_{n}=2^{\Omega\left(n\right)}$}.\textbf{ }

Like in the case of $\car\!\left(n\right)$ and $\sub\!\left(n\right)$,
the notion of $\D_{n}$ does not concern the size of $\mathbf{G}$,
thus providing a ground for artificial series of bad cases based on
growing alphabets. Our results show that actually the growing size
of $\mathbf{G}$ is not necessary: a strongly exponential lower bound
of $\D_{n}$ holds even if we restrict $\mathbf{G}$ to any nontrivial
fixed size.

\section{Conclusions and Future Work\label{sec:Conclusions-and-future}}

We have proved that both the considered thresholds (subset synchronization
of DFA and careful synchronization of PFA) are strongly exponential
even under two heavy restrictions (binary alphabets and strong connectivity).
We have improved the lower bounds of Martyugin, 2013 \citep{MAR5}.
However, the multiplicative constants in the exponents do not seem
to be the largest possible. 

For now there is no method giving upper bounds concerning the alphabet
size, so it may happen that binary cases are the hardest possible.
Such situation appears in the classical synchronization of DFA if
the \v{C}ern\'{y} Conjecture holds. 

From a more general viewpoint, our results give a partial answer to
the informal question: \emph{,,Which features of automata are needed
for obtaining strongly exponential thresholds?}'' However, for many
interesting restrictions we do not even know whether the corresponding
thresholds are superpolynomial. Namely, such restricted classes include
monotonic and aperiodic automata, cyclic and one-cluster automata,
Eulerian automata, commutative automata and others. For each such
class it is also an open question whether \noun{Subset synchronizability
}or \noun{Careful synchronizability }is solvable in polynomial time
with the corresponding restriction.

As it was noted before, for the general threshold $\car\!\left(n\right)$
there is a gap between $\mathcal{O}\!\left(n^{2}\cdot4^{\frac{n}{3}}\right)$
and $\Omega\!\left(3^{\frac{n}{3}}\right)$, which is subject to an
active research.

\bibliographystyle{ijfcs/ws-ijfcs}
\bibliography{C:/Users/Vojta/Desktop/SYNCHRO2/bib/ruco}

\end{document}